\documentclass[12pt]{article}
\usepackage{amsmath}
\usepackage{times}
\usepackage{bm}
\usepackage{natbib}
\usepackage[plain,noend]{algorithm2e}

\usepackage{float}
\usepackage{graphicx}
\usepackage{comment}
\graphicspath{ {figs/} }
\usepackage[english]{babel}
\usepackage{dsfont}
\usepackage[OT1]{fontenc}
\usepackage[utf8]{inputenc}
\usepackage[inline,shortlabels]{enumitem}
\usepackage{subfig}
\usepackage{lmodern}
\usepackage{color}
\usepackage{booktabs}
\usepackage{tikz}

\usepackage{refcount,nameref,zref-xr,zref-user} 
\zxrsetup{toltxlabel} 
\usepackage{url}

\usepackage{mathtools}
\usepackage{amsfonts}
\usepackage{amssymb}
\usepackage{bbm}

\makeatletter
\@ifclassloaded{biometrika}{%
  \makeatletter
  \renewcommand{\algocf@captiontext}[2]{#1\algocf@typo. \AlCapFnt{}#2}
  \def\@algocf@capt@plain{top}
  \renewcommand{\algocf@makecaption}[2]{%
    \addtolength{\hsize}{\algomargin}%
    \sbox\@tempboxa{\algocf@captiontext{#1}{#2}}%
    \ifdim\wd\@tempboxa >\hsize
      \hskip .5\algomargin%
      \parbox[t]{\hsize}{\algocf@captiontext{#1}{#2}}
    \else%
      \global\@minipagefalse%
      \hbox to\hsize{\box\@tempboxa}
    \fi%
    \addtolength{\hsize}{-\algomargin}%
  }
  \makeatother

  \addtolength\topmargin{35pt}

  \newcommand{\titlepaper}{Non-parametric efficient causal mediation with
  intermediate confounders}

  \jname{Biometrika}

  \markboth{ }{Non-parametric efficient causal mediation with intermediate
  confounders}
  \markboth{I.~D\'IAZ et al.}{Non-parametric efficient causal mediation with
  intermediate confounders}

  \newcommand{\authorlist}{
      \author{I.~D\'IAZ}
      \affil{Division of Biostatistics, Department of Population
        Health Sciences, Weill Cornell Medicine\\
      425 East 61\textsuperscript{st} Street, New York, NY 10065
      \email{ild2005@med.cornell.edu}}

      \author{N.S.~HEJAZI}
      \affil{Division of Epidemiology \& Biostatistics, School of Public
       Health, and Center for Computational Biology, University of California,
       Berkeley,\\ 2121 Berkeley Way, Berkeley, CA 94720-7360
      \email{nhejazi@berkeley.edu}}

      \author{K.E.~RUDOLPH}
      \affil{Department of Epidemiology, Mailman School of Public Health,
       Columbia University,\\
      ADDRESS
      \email{kr2854@cumc.columbia.edu}}

      \author{M.J.~VAN DER LAAN}
      \affil{Division of Epidemiology \& Biostatistics, School of Public
       Health, and Department of Statistics, University of California,
       Berkeley,\\ 2121 Berkeley Way, Berkeley, CA 94720-7360
      \email{laan@berkeley.edu}}
  }
}{%
  \usepackage{setspace}
  \usepackage{multirow}
  \usepackage[colorlinks,citecolor=blue,urlcolor=blue]{hyperref}
  \usepackage[margin=1in]{geometry}
  \usepackage{authblk}

  \usepackage{amsthm}

   \newtheorem{theorem}{Theorem}
   \newtheorem{lemma}{Lemma}

  {\theoremstyle{definition}}
  {\theoremstyle{definition}}

  \newcommand{\titlepaper}{Non-parametric efficient causal mediation
    with intermediate confounders}
  \newcommand{\authorlist}{
      \author[1]{Iv\'an D\'iaz \thanks{corresponding author:
          ild2005@med.cornell.edu}}
      \author[2,4]{Nima S.~Hejazi}
      \author[5]{Kara E.~Rudolph}
      \author[2,3,4]{Mark J.~van der Laan}
      \affil[1]{\small Division of Biostatistics, Department of Healthcare
           Policy \& Research, Weill Cornell Medicine}
      \affil[2]{\small Divsion of Epidemiology \& Biostatistics, School of
          Public Health, University of California, Berkeley}
      \affil[3]{\small Department of Statistics, University of
          California, Berkeley}
       \affil[4]{\small Center for Computational Biology,
          University of California, Berkeley}
      \affil[5]{\small Department of Epidemiology, Mailman School of
          Public Health, Columbia University}
  }
}
\makeatother

\AtEndDocument{\refstepcounter{theorem}\label{finalthm}}
\AtEndDocument{\refstepcounter{equation}\label{finaleq}}
\AtEndDocument{\refstepcounter{lemma}\label{finallemma}}

{\theoremstyle{definition}}

\DeclareMathOperator{\expit}{expit}

\DeclareMathOperator{\logit}{logit}
\DeclareMathOperator{\var}{\mathsf{Var}}

\renewcommand{\P}{\mathsf{P}}
\newcommand{\m}{\mathsf{b}}
\newcommand{\p}{\mathsf{p}}

\newcommand{\q}{\mathsf{q}}
\newcommand{\g}{\mathsf{g}}
\newcommand{\e}{\mathsf{h}}
\newcommand{\uu}{\mathsf{u}}
\newcommand{\h}{\mathsf{c}}
\newcommand{\vv}{\mathsf{v}}
\newcommand{\rr}{\mathsf{r}}

\newcommand{\indep}{\mbox{$\perp\!\!\!\perp$}}

\newcommand{\dd}{\mathrm{d}}
\newcommand{\Pn}{\mathsf{P}_n}
\newcommand{\Ps}{\mathsf{P_c}}

\newcommand{\thetatmle}{\hat\theta_{\mbox{\scriptsize tmle}}}
\newcommand{\thetaos}{\hat\theta_{\mbox{\scriptsize os}}}

\newcommand{\one}{\mathds{1}}

\newcommand{\E}{\mathsf{E}}

\DeclarePairedDelimiterX{\norm}[1]{\lVert}{\rVert}{#1}

\pgfdeclarelayer{background}
\pgfsetlayers{background,main}
\usetikzlibrary{arrows,positioning}
\tikzset{
>=stealth',
punkt/.style={
rectangle,
rounded corners,
draw=black, very thick,
text width=6.5em,
minimum height=2em,
text centered},
pil/.style={
->,
thick,
shorten <=2pt,
shorten >=2pt,}
}
\newcommand{\Vertex}[2]
{\node[minimum width=0.6cm,inner sep=0.05cm] (#2) at (#1) {$\footnotesize#2$};
}
\newcommand{\Vertexr}[2]
{\node[rectangle, draw, minimum width=0.6cm,inner sep=0.05cm] (#2) at (#1) {$\footnotesize#2$};
}
\newcommand{\ArrowR}[3]%
{ \begin{pgfonlayer}{background}
\draw[->,#3] (#1) to[bend right=30] (#2);
\end{pgfonlayer}
}
\newcommand{\ArrowL}[3]%
{ \begin{pgfonlayer}{background}
\draw[->,#3] (#1) to[bend left=45] (#2);
\end{pgfonlayer}
}
\newcommand{\EdgeL}[3]%
{ \begin{pgfonlayer}{background}
\draw[dashed,#3] (#1) to[bend right=-45] (#2);
\end{pgfonlayer}
}

\newcommand{\Arrow}[3]%
{ \begin{pgfonlayer}{background}
\draw[->,#3] (#1) -- +(#2);
\end{pgfonlayer}
}

\zexternaldocument*[supp:]{sm}
\title{\titlepaper}
\authorlist

\makeatletter
\@ifclassloaded{biometrika}{%
}{%
  \date{\today}
}
\makeatother

\begin{document}
\maketitle


\begin{abstract}
  Interventional effects for mediation analysis were proposed as a solution to
  the lack of identifiability of natural (in)direct effects in the presence of
  a mediator-outcome confounder affected by exposure. We present a theoretical
  and computational study of the properties of the interventional (in)direct
  effect estimands based on the efficient influence function (EIF) in the
  non-parametric statistical model. We use the EIF to develop two asymptotically
  optimal, non-parametric estimators that leverage data-adaptive regression for
  the estimation of nuisance parameters: a one-step estimator and a targeted
  minimum loss estimator. A free and open source \texttt{R} package implementing
  our proposed estimators is made available via GitHub. We further present
  results establishing the conditions under which these estimators are
  consistent, multiply robust, $n^{1/2}$-consistent and efficient. We illustrate
  the finite-sample performance of the estimators and corroborate our
  theoretical results in a simulation study. We also demonstrate the use of the
  estimators in our motivating application to elucidate the mechanisms behind
  the unintended harmful effects that a housing intervention had on adolescent
  girls' risk behavior.
\end{abstract}

\section{Introduction}

Treatment or exposure often affects an outcome of interest directly, or
indirectly by the mediation of some intermediate variables. Identifying and
quantifying the mechanisms underlying causal effects is an increasingly popular
endeavor in public health and the social sciences, as it can improve
understanding of both why and how treatments can be effective. The development
of tools for mediation analysis has occupied scientists across disciplines for
several decades. For example, the early work of~\citet{wright1921correlation,
wright1934method} on path analysis provided the foundations for the development
of some of the most commonly used mediation analysis
techniques~\citep{goldberger1972structural,baron1986moderator}. The advent of
novel causal inference frameworks such as non-parametric structural equation
models, directed acyclic graphs, and the
do-calculus~\citep{pearl1995causal,Pearl00}, as well as the approaches
of~\citet{Robins86}, \citet{spirtes2000causation}, \citet{dawid2000causal},
and~\cite{richardson2013single}, has uncovered important limitations of earlier
efforts focused on parametric structural equation models for mediation
analysis~\citep{pearl1998graphs,imai2010general}. Centrally, structural equation
models impose implausible assumptions on the data-generating mechanism, and are
thus of limited applicability to complex phenomena in biology, health,
economics, and the social sciences. For example, modern causal models have
revealed the failure of the widely popular method of~\citet{baron1986moderator}
in several important cases~\citep{cole2002fallibility}.

In recent decades, novel causal inference frameworks have allowed a number of
important developments in mediation analysis, including the non-parametric
decomposition of the average effect of a binary action into direct and indirect
effects~\citep{RobinsGreenland92,Pearl01}. The indirect effect quantifies the
effect on the outcome through the mediator and the direct effect quantifies the
effect through all other mechanisms. The identification of these natural
(in)direct effects relies on so-called \textit{cross-world} counterfactual
independencies, i.e., independencies on counterfactual variables indexed by
distinct hypothetical interventions. An important consequence of this necessary
assumption is that the natural (in)direct effect is not identifiable in
a randomized trial, which implies that scientific claims obtained from these
models are not falsifiable through
experimentation~\citep{Popper34,dawid2000causal,robins2010alternative}.
Furthermore, the required cross-world independencies are not satisfied in the
presence of mediator-outcome confounders that are affected by
treatment~\citep[hereby called intermediate confounders, see
][]{avin2005identifiability,tchetgen2014identification}. This represents
a fundamental limitation of the natural (in)direct effects, as such intermediate
variables are likely to be present in practice.

Recently, several techniques have been proposed to solve the aforementioned
problems. These methods may be divided in two classes. The first class attempts
to identify and estimate bounds on the natural (in)direct effects in the presence
of intermediate confounders~\citep[e.g.,][]{robins2010alternative,
tchetgen2014bounds, miles2015partial}, whereas the second class of methods is
concerned with alternative definitions of the (in)direct effects. For instance,
\citet{vansteelandt2012natural} argue that the natural (in)direct effects among
the treated are identified without cross-world counterfactual independencies. We
base our approach on a proposal first outlined by~\citet{petersen2006estimation}
and~\citet{van2008direct}, who argue that, in the absence of cross-world
assumptions, the standard identification formula for the natural direct effect
may be interpreted as a weighted average of controlled direct effects, with
weights given by a counterfactual distribution on the mediators. In subsequent
work, \citet{vanderweele2014effect} and \citet{vansteelandt2017interventional}
realized that the direct effect parameter of~\citet{petersen2006estimation}
corresponds to a decomposition of a total causal effect defined in terms of
a non-deterministic intervention on the mediator, and extended the methods to
the setting of intermediate confounders, calling the methods
\textit{interventional effects}~\citep{nguyen2019clarifying}. Specifically, the
interventional direct effect measures the effect of fixing the exposure at some
value $a'$ while drawing the mediator from those with exposure $a'$ (versus
$a^{\star}$), given covariates. The direct effect compares exposure $a'$ against
$a^{\star}$, with the mediator in both cases randomly drawn from those with
exposure $a^{\star}$, given covariates. \cite{zheng2017longitudinal} study the
identification and estimation of an interventional effect in which the mediator
is drawn from its counterfactual distribution conditional on baseline variables
and the counterfactual intermediate confounder. Interventional effects have also
been studied by~\citet{lok2016defining,lok2019causal},
\citet{rudolph2017robust}, and~\citet{didelez2006direct}, among several others.

Our contribution to the literature is a study of the optimality properties of
the estimand of the interventional (in)direct effects. We derive the efficient
influence function and use it to develop $n^{1/2}$-consistent non-parametric
estimators for the interventional effects of~\citet{vanderweele2014effect}.
Importantly, our estimators (i) can incorporate intermediate confounders, (ii)
can incorporate flexible data-adaptive regression for nuisance parameters, (iii)
can accommodate high-dimensional mediators, (iv) can achieve the non-parametric
efficiency bound, and (v) are multiply robust to estimation of combinations of
nuisance parameters. To the best of our knowledge, this is the first proposal of
such a method. The estimators proposed by~\citet{van2008direct}
and~\citet{zheng2012targeted} may be used to estimate interventional effects,
and satisfy properties (ii)-(v) but not (i). The weighting estimators
of~\citet{vanderweele2014effect} and the regression estimators
of~\citet{vansteelandt2017interventional} satisfy property (i) but not (ii)-(v).
In particular, we say that these estimators do not satisfy (iii) because they
would require estimation of a possibly high-dimensional conditional density on
the mediator. The estimators of \cite{rudolph2017robust} satisfy properties (i),
(ii), and (v), but not (iii) and (iv). The direct effects
of~\cite{zheng2017longitudinal} are not directly applicable to our motivating
example, as they measure causal pathways different from the paths of interest in
our motivating application (see Section~\ref{supp:remark:wen} in the
supplementary materials). Concurrent to the
development of this work, \citet{benkeser2020nonparametric} studied the
non-parametric efficiency bound for interventional effects and proposed
non-parametric efficient estimators that can leverage data-adaptive regression.
Unlike the estimators we propose, those of~\citet{benkeser2020nonparametric}
also require estimation of possibly high-dimensional conditional densities of
the mediator. Regarding property (iv), we note that competing estimators relying
on parametric assumptions can deliver better efficiency if the model-based
restrictions imposed are correct, but are otherwise inconsistent.

We develop two estimators inspired by the efficient influence function:
a one-step estimator that can accommodate data structures where at most one of
the mediator or intermediate confounder is high-dimensional, and a targeted
minimum loss (TML) estimator under the assumption that the intermediate
confounder is binary, such as in our motivating application. In addition to
aiding in the development of the estimators, our studies of the parameter
functional in a non-parametric model provide new insights regarding efficiency
and multiple robustness achievable in estimation of interventional effects. In
particular, we derive the non-parametric efficiency bound and discuss the
robustness properties of estimators based on the efficient influence function.
Finally, we provide a free and open source software implementation of our
proposed methodology, the \texttt{medoutcon}~\citep{hejazi2020medoutcon}
package, for the \texttt{R} language and environment for statistical
computing~\citep{R}.

\section{Notation and definition of (in)direct effects}

Let $A$ denote a categorical treatment variable, let $Y$ denote a continuous or
binary outcome, let $M$ denote a multivariate mediator, let $W$ denote a vector
of observed pre-treatment covariates, and let $Z$ denote a multivariate
mediator-outcome confounder affected by treatment. Our setup will allow for at
most one of $M$ and $Z$ to be high-dimensional, but not both. Let $O = (W, A,
Z, M, Y)$ represent the observed data, and let $O_1, \ldots, O_n$ denote a
sample of $n$ i.i.d.~observations of $O$. We formalize the definition of our
counterfactual variables using the following non-parametric structural equation
model (NPSEM), though equivalent methods may be developed by taking the
counterfactual variables as primitives. Assume the data-generating process
satisfies
\begin{align}\label{eq:npsem}
  W &= f_W(U_W); A = f_A(W, U_A); Z=f_Z(W, A, U_Z); \nonumber \\
  M &= f_M(W, A, Z, U_M); Y = f_Y(W, A, Z, M, U_Y).
\end{align}
In our illustrative example, we have the exclusion restrictions $M = f_M(W, Z,
U_M)$ and $Y = f_Y(W, Z, M, U_Y)$ --- despite this, we present general
methodology compatible with the model (\ref{eq:npsem}). Here, $U=(U_W,
U_A,\allowbreak U_Z, U_M, U_Y)$ is a vector of exogenous factors, and the
functions $f$ are assumed deterministic but unknown. We use $\P$ to denote the
distribution of $O$, and $\Ps$ to denote the distribution of $(O,U)$. We let the
true distribution of the data, $\P$, be an element of the non-parametric
statistical model defined as all continuous densities on $O$ with respect to
some dominating measure $\nu$. Let $\p$ denote the corresponding probability
density function, and $\mathcal{W, Z, M}$ denote the range of the respective
random variables. We let $\E$ and $\E_c$ denote corresponding expectation
operators, and define $\P f = \int f(o)\dd \P(o)$ for a given function $f(o)$.
We use $\Pn$ to denote the empirical distribution of $O_1, \ldots, O_n$.

In our discussion, we define several additional useful parameterizations. We use
$\g(a \mid w)$ to denote the true probability mass function of $A=a$ conditional
on $W = w$, and $\e(a \mid m, w)$ to denote the probability mass function of
$A=a$ conditional on $(M, W)=(m,w)$. We use $\m(a,z,m,w)$ to denote the true
outcome regression function $\E(Y \mid A = a,Z = z,M = m, W = w)$, and $\q(z
\mid a,w)$ and $\rr(z \mid a,m,w)$ to denote the corresponding true conditional
densities of $Z$. For a random variable $X$, we let $X_a$ denote the
counterfactual outcome observed in a hypothetical world in which $\P(A=a)=1$.
For example, we have $Z_a = f_M(W,a,U_Z)$, $M_a=f_M(W,a,Z_a,U_M)$, and
$Y_a=f_Y(W,a,Z_a,M_a,U_Y)$. Likewise, we let $Y_{a,m} =f_Y(W,a,Z_a,m,U_Y)$
denote the value of the outcome in a hypothetical world where $\P(A=a,M=m)=1$.
It will be useful sometimes to consider a representation of model
(\ref{eq:npsem}) in terms of the directed acyclic graph (DAG) in
Figure~\ref{fig:dag}.
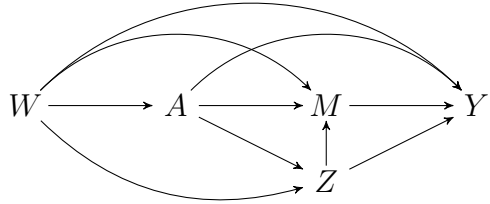
\begin{figure}[!htb]
  \centering
  \begin{tikzpicture}
    \Vertex{0, -1}{Z}
    \Vertex{-4, 0}{W}
    \Vertex{0, 0}{M}
    \Vertex{-2, 0}{A}
    \Vertex{2, 0}{Y}
    \ArrowR{W}{Z}{black}
    \Arrow{Z}{M}{black}
    \Arrow{W}{A}{black}
    \Arrow{A}{M}{black}
    \Arrow{M}{Y}{black}
    \Arrow{A}{Z}{black}
    \Arrow{Z}{Y}{black}
    \ArrowL{W}{Y}{black}
    \ArrowL{A}{Y}{black}
    \ArrowL{W}{M}{black}
  \end{tikzpicture}
  \caption{Directed Acyclic Graph.}
  \label{fig:dag}
\end{figure}

\subsection{Interventional mediation effects}

In this work, we define the total effect of $A$ on $Y$ in terms
of a contrast between two user-given values $a', a^{\star} \in \mathcal
A$. Figure \ref{fig:dag} reveals four paths involved in this effect,
namely $A\rightarrow Y$, $A\rightarrow M \rightarrow Y$,
$A\rightarrow Z \rightarrow Y$, and
$A\rightarrow Z \rightarrow M \rightarrow Y$ . One approach to
mediation analysis considers the \emph{natural direct effect
(NDE)} and the \emph{natural indirect effect (NIE)}, defined as
$\E_c(Y_{a',M_{a^{\star}}} - Y_{a^{\star}, M_{a^{\star}}})$ and
$\E_c(Y_{a',M_{a'}} - Y_{a',M_{a^{\star}}})$, respectively
\citep{RobinsGreenland92,Pearl01}. In terms of the NPSEM
(\ref{eq:npsem}), these effects are defined as
\begin{align*}
  \text{NDE} &= \E_c\{f_Y(W,a',Z_{a'},M_{a^{\star}},U_Y) -
               f_Y(W,a^{\star},Z_{a^{\star}},M_{a^{\star}},U_Y)\}\\
  \text{NIE} &= \E_c\{f_Y(W,a',Z_{a'},M_{a'},U_Y) -
               f_Y(W,a',Z_{a'},M_{a^{\star}},U_Y)\}.
\end{align*}
Inspection of the above definitions reveals that the natural direct effect
measures the effect through paths \emph{not} involving the mediator
($A\rightarrow Y$ and $A\rightarrow Z \rightarrow Y$), whereas the natural
indirect effect measures the effect through paths involving the mediator
($A\rightarrow M \rightarrow Y$ and $A\rightarrow Z \rightarrow M \rightarrow
Y$). This effect definition is appealing because the sum of the natural direct
and indirect effects equals the average treatment effect $\E_c(Y_1-Y_0)$.
Unfortunately, natural direct and indirect effects are not generally identified
in the presence of a mediator-outcome confounder affected by treatment such as
in the DAG in Figure
(\ref{fig:dag})~\citep{avin2005identifiability,tchetgen2014identification}. The
reason beyond the lack of identification is that the so-called cross-world
counterfactual independence $Y_{a', m}\indep M_{a^{\star}} \mid W$, necessary
for identification, cannot be satisfied in the presence of intermediate
confounders. Intuitively, this is because the interventions $A=a'$ and
$A=a^{\star}$ yield two counterfactual variables $Z_{a'}$ and $Z_{a^{\star}}$
which are not independent and enter the structural equation definition of
$Y_{a', m}$, and $M_{a^{\star}}$, respectively.

To solve this problem while retaining the path decomposition employed by the
natural direct and indirect effects, we adopt an approach previously
outlined~\citep{petersen2006estimation,van2008direct,
zheng2012targeted,vanderweele2014effect,rudolph2017robust}, defining direct and
indirect effects using stochastic interventions on the mediator. Let $G_a$
denote a random draw from the conditional distribution of $M_a$, conditional on
$W$. Consider the effect of $A$ on $Y$ defined as the difference in expected
outcome in hypothetical worlds in which $(A,M) = (a', G_{a'})$ versus $(A,M)
= (a^{\star}, G_{a^{\star}})$ with probability one, which may be decomposed into
direct and indirect effects as follows
\begin{equation}
\E_c(Y_{a', G_{a'}} - Y_{a^{\star}, G_{a^{\star}}}) =
      \underbrace{\E_c(Y_{a', G_{a'}} - Y_{a', G_{a^{\star}}})}_{\text{Indirect
          effect (through $M$)}} +
  \underbrace{\E_c(Y_{a', G_{a^{\star}}} - Y_{a^{\star},
      G_{a^{\star}}})}_{\text{Direct effect (not through $M$)}}
\label{eq:decomp}.
\end{equation}
Like the natural direct effect, this direct effect measures the effects through
paths not involving the mediator. Likewise, the indirect effect measures the
effect through paths involving the mediator. However, natural and interventional
mediation effects have different interpretations. The interventional indirect
effect measures the effect of fixing the exposure at $a'$ while setting the
mediator to a random draw $G_{a^{\star}}$ from those with exposure $a'$ versus
a random draw $G_{a'}$ from those with exposure $a^{\star}$, given covariates.
The interventional direct effect compares exposure $a'$ versus $a^{\star}$ with
the mediator in both cases set to a random draw $G_{a^{\star}}$ from those with
exposure $a^{\star}$, given covariates. This is different in nature to the
natural (in)direct effects. Specifically, the natural indirect effect measures
the effect of fixing the exposure at $a'$ while setting the mediator to the
value $Z_{a'}$ it would have taken under $a'$ for that unit, to the value
$Z_{a^{\star}}$ it would have taken under $a^{\star}$. The natural direct effect
compares exposure $a'$ versus $a^{\star}$ with the mediator in both cases fixed
to the unit-specific value $Z_{a^{\star}}$ it would have taken under
$a^{\star}$. Notably, the total effect obtained as the sum of the interventional
direct and indirect effects is different from the average treatment effect.
Further discussion of the interpretation of these effects may be found
in~\citet{vanderweele2014effect, vansteelandt2017interventional,
nguyen2019clarifying}.

Identification and optimal estimation of a similar effect
decomposition is investigated by~\citet{zheng2017longitudinal}. We
briefly discuss the relation of our work with their approach in
Section~\ref{supp:remark:wen} in the supplementary materials.

In the sequel, we focus on identification and estimation of $\theta_c
= \E_c(Y_{a', G_{a^{\star}}})$, from which the effect decomposition in
(\ref{eq:decomp}) can be obtained. \cite{vanderweele2014effect} show that, under
the assumptions
\begin{enumerate*}[label=(\roman*)]
\item $Y_{a,m}\indep A\mid W$,\label{ass:ncay}
\item $M_{a}\indep A\mid W$,\label{ass:ncam} and
\item $Y_{a,m}\indep M\mid (A,W,Z)$,\label{ass:ncmy}
\end{enumerate*}
$\theta_c$ is identified and is equal to
\begin{equation}
\theta = \int \m(a',z,m,w)\q(z\mid
  a',w)\p(m\mid a^{\star},w)\p(w)\dd \nu(w,z,m).\label{eq:thetadef}
\end{equation}
Assumption \ref{ass:ncay} states that, conditional on $W$, there is no
unmeasured confounding of the relation between $A$ and $Y$;
assumption~\ref{ass:ncam} states that conditional on $W$ there is no unmeasured
confounding of the relation between $A$ and $M$; and~\ref{ass:ncmy} states that
conditional on $(A,W,Z)$ there is no unmeasured confounding of the relation
between $M$ and $Y$. Next, we turn our attention to a discussion of efficiency
theory for estimation of the statistical parameter $\theta$, which depends only
on the observed data distribution $\P$.

\section{Optimal estimation of $\theta$ in the non-parametric model}

The \textit{efficient influence function} (EIF) is a key object in general
semi-parametric estimation theory, as it characterizes the asymptotic behavior
of all regular and efficient estimators~\citep{Bickel97,van2002part}. The EIF is
often useful in constructing locally efficient estimators. Some of the most
common approaches for this are (i) using the EIF as an estimating
equation~\citep[e.g.,][]{vanderLaan2003}, (ii) using the EIF in a one-step bias
correction~\citep[e.g.,][]{pfanzagl1982contributions}, and (iii) using the EIF
to construct targeted minimum loss estimators~\citep{vdl2006targeted,
vanderLaanRose11,vanderLaanRose18}.
The EIF estimating equation often enjoys desirable properties such as multiple
robustness, which allows for some components of the data distribution to be
inconsistently estimated while preserving consistency of the estimator. In
addition, the asymptotic analysis of estimators constructed using the EIF often
yields second-order bias terms, which require slow convergence rates (e.g.,
$n^{-1/4}$) for the nuisance parameters involved, thereby enabling the use of
flexible regression techniques in estimating these quantities.

\begin{theorem}[Efficient influence function]\label{theo:eif}
  For fixed $a'$, $a^{\star}$ define
\begin{equation}
  \begin{split}
    \h(a, z, m, w) & = \frac{\p(m\mid a^{\star}, w)}{\p(m\mid a, z, w)}\\
    \uu(z,a,w) &=\int_{\mathcal
      M}\m(a,z,m,w)\p(m\mid a^{\star},w)\dd\nu(m)\\
    \vv(a,w)&=\int_{\mathcal
      M\times Z}\m(a',z,m,w)\q(z\mid
    a',w)\p(m\mid a,w)\dd\nu(m,z).
  \end{split}\label{eq:defhuv}
\end{equation}
The efficient influence function for $\theta$ in the nonparametric model $M$
is equal to $D_\P(o)-\theta$, where
\begin{align}
  D_\P(o) &= \frac{\one\{a=a'\}}{\g(a'\mid w)}\h(a',z,m,w)\{y - \m(a',z,m,w)\}\label{eq:DY}\\
            & + \frac{\one\{a=a'\}}{\g(a'\mid w)}\left\{\uu(z,a',w)-\int_{\mathcal
              Z}\uu(z,a',w)\q(z\mid a',w)\dd\nu(z)\right\}\label{eq:Du}\\
            & + \frac{\one\{a=a^{\star}\}}{\g(a^{\star}\mid w)}\left\{\int_{\mathcal
              Z}\m(a',z,m,w)\q(z\mid a',w)\dd\nu(z)-\vv(a^{\star},
              w)\right\}\label{eq:Dv}\\
            & + \vv(a^{\star},w)\notag.
\end{align}
\end{theorem}

Inspection of this theorem reveals that computation of the EIF requires
estimation of possibly high-dimensional densities on the mediator $M$ and
confounder $Z$, as well as integrals with respect to these densities. This poses
an important challenge due to the curse of dimensionality. Fortunately, if
either one of $M$ or $Z$ is low-dimensional, we can overcome this problem using
an alternative parameterization of the densities. As $Z$ is binary in our
motivating application, we assume in the following that $Z$ is low-dimensional,
allowing its conditional density to be easily estimated, while we allow $M$ to
be high-dimensional. The case of high-dimensional $Z$ and low-dimensional $M$ is
treated in the \href{sm}{Supplementary Materials}. The EIF given in
Theorem~\ref{theo:eif} may be represented in terms of the expressions given in
Lemma~\ref{lemma:aeif}, which does not depend on conditional densities or
integrals on the mediator.

\begin{lemma}[Alternative representation of the EIF for
  high-dimensional $M$]\label{lemma:aeif}
The functions $\h$, $\uu$, and $\vv$ admit the following alternative
representations:
  \begin{align}
    \h(a, z, m, w)&=\frac{\g(a\mid w)}{\g(a^{\star}\mid w)}
                    \frac{\q(z\mid a,w)}{\rr(z\mid a,m,w)}
      \frac{\e(a^{\star}\mid m, w)}{\e(a\mid m,w)}\label{eq:param}\\
    \uu(z,a,w) &= \E\left\{\m(A,Z,M,W)\h(A,Z,M,W),\bigg|\,
      Z=z,A=a,W=w\right\},\label{eq:phi}\\
    \vv(a,w) &= \E\left\{\int_{\mathcal
               Z}\m(a',z,M,W)\q(z\mid a',W)\dd\nu(z)\,\bigg|\, A=a,W=w\right\}.
   \label{eq:lambda}
  \end{align}
\end{lemma}
In the following, we denote $\eta=(\m,\g,\e,\q,\rr,\uu,\vv)$ and $D_\P(o)
= D_\eta(o)$. This choice of parameterization has important consequences for the
purpose of estimation, as it helps to bypass estimation of the (possibly
high-dimensional) conditional density of the mediators, instead allowing for the
use of regression methods, which are far more commonly found in the statistics
literature and software, to be used for estimation of the relevant quantities.
We note that this approach will still require the estimation of possibly
multivariate or continuous densities on $Z$. If $Z$ is continuous, its density
may be estimated using parametric methods, data-adaptive techniques such as
those developed by~\citet{vapnik2000svm, Diaz11, izbicki2017converting}, or
ensembles of the previous methods~\citep{vanderLaanDudoitvanderVaart06}. If $Z$
is a multivariate vector, its density may be estimated by sequentially
conditioning on each component of the vector, e.g., $Z_1 \mid Z_2, Z_3,\ldots$,
$Z_2 \mid Z_3, Z_4,\ldots$, and so on, although we expect this method to work
reliably only in low-dimensional settings.

In the sequel, we will use $\eta_1=(\m_1,\g_1,\e_1,\q_1,\rr_1,\uu_1,\vv_1)$ to
denote any value of the nuisance parameters, possibly different from the true
value $\eta$ (e.g., the probability limit of an estimator $\hat\eta$). In
addition to the expression for the EIF in Lemma~\ref{lemma:aeif}, it is
important to understand the behavior of the difference $\P D_{\eta_1} - \theta$,
which is expected to yield a second-order term in differences $\eta_1-\eta$, so
that consistent estimation of $\theta$ is possible under consistent estimation
of certain configurations of the parameters in $\eta$. As we will see in
Theorems~\ref{theo:asos} and~\ref{theo:astmle}, this second-order term is
fundamental in the construction of consistent asymptotically normal (CAN)
estimators. Theorem~\ref{supp:theo:mr}, found in the \href{sm}{Supplementary
Materials}, shows this second-order term. The following lemma is a direct
consequence.
\begin{lemma}[Multiple robustness of $D_\eta(O)$]\label{lemma:robust}
  Let $\eta_1=(\m_1,\g_1,\e_1,\q_1,\rr_1,\uu_1,\vv_1)$ be such that one
  of the following conditions hold:
  \begin{enumerate}[label=(\roman*)]
  \item $\vv_1=\vv$ and either $(\q_1,\e_1,\rr_1)=(\q,\e,\rr)$
    or $(\m_1,\q_1)=(\m,\q)$ or $(\m_1,\uu_1)=(\m,\uu)$, or
  \item $\g_1=\g$ and either $(\q_1,\e_1,\rr_1)=(\q,\e,\rr)$
    or $(\m_1,\q_1)=(\m,\q)$ or $(\m_1,\uu_1)=(\m,\uu)$.
  \end{enumerate}
  Then $\P D_{\eta_1} = \theta$ with $D_\eta$ defined as in
  Theorem~\ref{lemma:aeif}.
\end{lemma}
Lemma~\ref{lemma:robust} implies that it is possible to construct consistent
estimators for $\theta$ under consistent estimation of the nuisance parameters
in $\eta$ in the configurations described therein. We note that the cases
$(\m_1,\vv_1,\uu_1)=(\m,\vv,\uu)$ and $(\m_1,\g_1,\uu_1)=(\m,\g,\uu)$ may be
uninteresting if the re-parametrization in Lemma~\ref{lemma:aeif} is used to
estimate the EIF, since, in that case, consistent estimation of $\uu$ and $\vv$
will generally require consistent estimation of $(\m,\q,\rr,\e)$, in addition to
the outer conditional expectations in Equations (\ref{eq:phi}) and
(\ref{eq:lambda}). In the next section, we discuss the construction of
$n^{1/2}$-consistent non-parametric estimators using one-step and targeted
minimum loss estimators.

\section{Efficient estimation}\label{est}

We will use $\hat \eta$ to denote an estimator of
$\eta=(\m,\g,\e,\q,\rr,\uu,\vv)$, noting that all of these parameters are
conditional expectations which may be estimated using data-adaptive regression
techniques or ensembles thereof. For estimation of $\uu$ and $\vv$, the
estimators of $(\m,\g,\q,\e)$ may be plugged in to the outcome variable defined
in Equations (\ref{eq:phi}) and (\ref{eq:lambda}). Then, any data-adaptive
regression method may be used to estimate the outer expectation. For example,
$\uu(z,a',w)$ may be estimated by computing the auxiliary covariate
$\hat\m(A,Z,M,W)\hat \h(A,Z,M,W)$, regressing it on $(Z,A,W)$, and evaluating
the predictor at $(z,a',w)$.


In order to avoid imposing Donsker conditions on the initial estimators, we will
use cross-fitting~\citep{klaassen1987consistent,zheng2011cross,
chernozhukov2016double} in the estimation procedure. Let ${\cal V}_1, \ldots,
{\cal V}_J$ denote a random partition of the index set $\{1, \ldots, n\}$ into
$J$ prediction sets of approximately the same size. That is, ${\cal V}_j\subset
\{1, \ldots, n\}$; $\bigcup_{j=1}^J {\cal V}_j = \{1, \ldots, n\}$; and ${\cal
V}_j\cap {\cal V}_{j'} = \emptyset$. In addition, for each $j$, the associated
training sample is given by ${\cal T}_j = \{1, \ldots, n\} \setminus {\cal
V}_j$. We denote by $\hat \eta_{j}$ the estimator of $\eta$, obtained by
training the corresponding prediction algorithm using only data in the sample
${\cal T}_j$. Further, we let $j(i)$ denote the index of the validation set
containing observation $i$.

\subsection{One-step estimator}\label{one_step}

The one-step estimator is defined as the solution to the cross-fitted efficient
influence function estimating equation:
\begin{equation}\label{eq:aipw}
  \thetaos = \frac{1}{n} \sum_{i = 1}^n D_{\hat\eta_{j(i)}}(O_i).
\end{equation}
Empirical process theory may be used to derive the asymptotic distribution of
$\thetaos$. Asymptotic normality and efficiency of the estimator is detailed in
the following theorem:
\begin{theorem}[Weak convergence of the one-step estimator]\label{theo:asos} Let
  $\lVert \cdot \rVert$ denote the $L_2(\P)$ norm defined as
  $\lVert f \rVert^2=\int f^2\dd\P$. Assume
  \begin{enumerate}[label=(\roman*)]
  \item \label{ass:pos} Positivity of exposure and confounder probabilities:
    \[\P\{\g(a^{\star} \mid W) > \epsilon\}=\P\{\rr(Z\mid a',M,W) >
    \epsilon\}=\P\{\e(a'\mid M,W) > \epsilon\}=1\] for some
    $\epsilon > 0$.
  \item \label{ass:so} $n^{1/2}$-convergence of second order terms:
    \[\lVert \hat\vv-\vv \rVert\, \lVert \hat\g-\g \rVert +
      \lVert \hat\m-\m \rVert \{\lVert \hat\q-\q \rVert +
      \lVert \hat\e - \e \rVert + \lVert \hat\rr - \rr \rVert\} +
      \lVert \hat\uu - \uu \rVert\, \lVert \hat\q-\q \rVert =
      o_\P(n^{-1/2}).\]
  \end{enumerate}
Then, we have
\[n^{1/2}(\thetaos-\theta)\rightsquigarrow N(0,\sigma^2),\] where
$\sigma^2=\var\{D_\eta(O)\}$ is the non-parametric efficiency bound.
\end{theorem}

Condition~\ref{ass:pos} of the theorem is standard in causal inference, simply
stating that there is enough experimentation in the data-generating mechanism
such that there are no deterministic relations between the variables in the left
and right hand sides of the conditioning statement. We note in particular that
the positivity assumption regarding the density of the mediator $\rr$ is likely
to be violated if $Z$ is high-dimensional. Condition~\ref{ass:so} can be
satisfied, for example, if all the nuisance parameters involved converge to
their true values at $n^{1/4}$-rate in $L_2(\P)$-norm. This convergence rate is
much slower than the typical $n^{1/2}$ parametric rate. To increase assurance
that this assumption is satisfied, we depart from the classical but unrealistic
parametric setting by allowing usage of flexible data-adaptive estimators from
the machine and statistical learning literatures. The required rates are
achievable by many data-adaptive regression algorithms. See, for example,
\citet{bickel2009simultaneous} for rate results on $\ell_1$ regularization,
\citet{wager2015adaptive} for rate results on regression trees, and
\citet{chen1999improved} for neural networks. The assumption may also be
satisfied by the highly adaptive lasso \citep[HAL,][]{benkeser2016highly} under
the condition that the true regression functions are right-hand continuous with
left-hand limits and have sectional variation norm bounded by a constant.
Ensemble learners such as the Super Learner~\citep{vdl2007super} may also be
used. Super learning builds a weighted combination of predictors in a user-given
library of candidate estimators, where the weights are assigned so as to
minimize the cross-validated risk of the resulting combination; moreover, it has
important theoretical guarantees~\citep{vanderLaanDudoitvanderVaart06,
vanderVaartDudoitvanderLaan06}, such as asymptotic equivalence to the oracle
selector. We also note that the use of cross-fitting avoids the Donsker
conditions that are typically required when the nuisance parameters are
estimated with data-adaptive methods~\citep{klaassen1987consistent,
zheng2011cross, chernozhukov2016double}.

The weak convergence established in Theorem~\ref{theo:asos} is useful to derive
confidence intervals. Under the assumptions of the theorem, an estimator
$\hat\sigma^2$ of $\sigma^2$ may be obtained as the empirical variance of
$D_{\hat\eta_{j(i)}}(O_i)$, allowing a Wald-type confidence interval to be
constructed as $\thetaos\pm z_{1-\alpha/2} \cdot (\hat\sigma/\sqrt{n})$.

While the one-step estimator has optimal asymptotic performance, its
finite-sample behavior may be adversely affected by the inverse probability
weighting involved in computation of $D_{\hat\eta}(O_i)$. In particular,
$\thetaos$ is not guaranteed to remain within the bounds of the parameter space.
This issue may be attenuated by performing weight stabilization. The estimated
EIF $D_{\hat\eta_{j(i)}}(O_i)$ can be weight-stabilized by dividing term
(\ref{eq:DY}) by the empirical mean of $\one\{A_i=a'\}/\hat \g_{j(i)}(A_i\mid
W_i)\times \hat \h_{j(i)}^{\star}(A_i,Z_i,M_i,W_i)$, dividing term (\ref{eq:Du})
by the empirical mean of $\one\{A_i=a'\}/\hat \g_{j(i)}(A_i\mid W_i)$, and
dividing term (\ref{eq:Dv}) by the empirical mean of $\one\{A_i=a^{\star}\}/
\hat \g_{j(i)}(A_i\mid W_i)$. Theorem~\ref{theo:asos} remains valid under this
weight stabilization, and inference may be performed using the normal
distribution and the standard error estimator based on the EIF. The targeted
minimum loss (TML) estimation framework provides a more principled way to obtain
estimators that remain in the parameter space. The TML estimator is constructed
by tilting an initial data-adaptive estimator $\hat\eta$ towards a solution
$\tilde\eta$ of the estimating equation $\Pn D_{\tilde\eta}
= \theta(\tilde\eta)$, where $\theta(\tilde\eta)=\thetatmle$ is the resulting
substitution estimator obtained by plugging in the estimates $\tilde\eta$ in the
parameter definition (\ref{eq:thetadef}). Since it is a substitution estimator,
a TML estimator respects the bounds of the parameter space by definition. The
fact that the nuisance estimators solve the relevant estimating equation is used
to obtain a weak convergence result analogous to Theorem~\ref{theo:asos}. Thus,
while the TML estimator is expected to attain the same optimal asymptotic
behavior as the one-step estimator, its finite-sample behavior may be better by
virtue of its being a substitution estimator.

\subsection{Targeted minimum loss estimator}\label{tmle}

We focus the presentation of the TML estimator on the case of a binary variable
$Z$ and an outcome $Y\in[0,1]$, such as in our motivating example. If the
outcome is bounded with known bounds $[a,b]$, then a simple transformation
yields $Y\in[0,1]$. For binary $Z$, the EIF further simplifies as
\begin{align}
  D_\eta(o) - \vv(a^{\star},w)&= \frac{\one\{a=a'\}}
  {\g(a'\mid w)}\h(a',z,m,w)\{y - \m(a',z,m,w)\}\label{sc:Y}\\
                        & + \frac{\one\{a=a'\}}{\g(a'\mid w)}\{\uu(1,a',w)-
                        \uu(0,a',w)\}\left\{z -
                          \q(1\mid a',w)\right\}\label{sc:Z}\\
                        & + \frac{\one\{a=a^{\star}\}}{\g(a^{\star}
                        \mid w)}\left\{\sum_{z=0}^1\m(a',z,m,w)\q(z
                        \mid a',w)-\vv(a^{\star},w)\right\}\label{sc:M},
\end{align}
where we note that terms (\ref{sc:Y}), (\ref{sc:Z}), and (\ref{sc:M}) are score
equations of the type $H\{X - \E(X\mid L)\}$ for appropriately defined random
variables $H$, $X$, and $L$. A cross-fitted TML estimator~\citep{zheng2011cross}
may thus be computed in the following steps by iterating through solutions of
these score equations in logistic tilting models:
\begin{enumerate}[label=Step \arabic*., align=left, leftmargin=*]
\item Initialize $\tilde\eta =\hat\eta$.
\item \label{step:computeH} For each subject, compute the auxiliary covariates
  \begin{align*}
    H_{Y,i} &= \frac{\one\{A_i=a'\}}{\tilde\g_{j(i)}(a^{\star}\mid
              W_i)}\tilde \h^{\star}(A_i,Z_i,M_i,W_i)\\
    H_{Z,i} & = \frac{\one\{A_i=a'\}}{\tilde\g_{j(i)}(a'\mid W_i)}
    \{\tilde\uu_{j(i)}(1,A_i,W_i)-\tilde\uu_{j(i)}(0,A_i,W_i)\}
  \end{align*}
\item \label{step:fit} Fit the logistic tilting models
  \begin{align*}\logit \m_\beta(A_i,Z_i,M_i,W_i) &= \logit \widetilde
    \m_{j(i)}(A_i,Z_i,M_i,W_i) + \beta_Y H_{Y,i},\\
    \logit \q_\beta(1\mid A_i,W_i) &= \logit \tilde
    \q_{j(i)}(1\mid A_i,W_i) + \beta_Z H_{Z,i},
  \end{align*}
  where $\logit(p) = \log\{p(1-p)^{-1}\}$. Here, $\logit \widetilde\m(a,z,m,w)$
  and $\logit \tilde\q(a,w)$ are offset variables (i.e., a variable with known
  parameter value equal to one). The parameter $\beta$ may be estimated by
  running standard logistic regression models. For example, $\beta_Y$ may be
  estimated by running a logistic regression of $Y_i$ on $H_Y(A_i,Z_i,M_i,W_i)$
  with no intercept term and an offset term equal to $\logit
  \widetilde\m_{j(i)}(A_i,Z_i,M_i,W_i)$. Let $\hat\beta$ denote the estimate,
  and let $\widetilde \m=\m_{\hat\beta}$ and $\tilde \q=\q_{\hat\beta}$ denote
  the updated estimates.
\item \label{step:itera} Repeat \ref{step:computeH}-\ref{step:fit} until
  convergence. We stop the iteration when the score equations are solved up to
  a factor of $\{\sqrt{n}\log(n)\}^{-1}$ multiplied by their standard deviation.
\item Letting $(\widetilde \m, \tilde \q)$ denote the estimators in the
  last step of the above iteration, compute the pseudo-outcome
  $\widetilde Y$ as \[\widetilde Y_i = \sum_{z \in \{0, 1\}}
  \widetilde \m_{j(i)}(a',z,M_i,W_i) \tilde \q_{j(i)}(z\mid a',W_i).\]
  Let $\tilde \vv_{j(i)}(a,w)$ denote a cross-fitted estimator constructed by
  regressing $\widetilde Y$ on $(A,W)$, using data-adaptive methods. Compute the
  auxiliary variable \[H_{M,i}=\frac{\one\{A_i=a^{\star}\}}{\tilde
  \g_{j(i)}(A_i\mid W_i)}\]
\item Fit the logistic tilting model \[\logit \vv_\beta(A_i,W_i) = \logit
  \tilde \vv_{j(i)}(A_i,W_i) + \beta_M H_{M,i},\] by running a logistic
  regression of $\widetilde Y_i$ on $H_M(A_i,W_i)$ with without an intercept and
  with an offset term equal to $\logit \tilde\vv_{j(i)}(A_i,W_i)$. Let
  $\tilde\vv=\vv_{\hat\beta}$ denote the updated estimator. The TML estimator of
  $\theta$ is then defined \[\thetatmle = \frac{1}{n}\sum_{i=1}^n\tilde
  \vv_{j(i)}(a^{\star},W_i).\]
\end{enumerate}
Note that we stop the iteration once the relevant score equations are solved up
to a factor of $\{\sqrt{n}\log(n)\}^{-1}$. This ensures that the error in
solving the estimating equations is $o_P(n^{-1/2})$, so that
Theorem~\ref{theo:astmle} holds, while also improving small sample behavior by
avoiding unnecessary extra fitting. The large sample distribution of our
TML estimator is given in the following theorem:
\begin{theorem}[Weak convergence of the TML estimator]\label{theo:astmle}
  Assume $Z$ is binary. Assume \ref{ass:pos} and \ref{ass:so} as in
  Theorem~\ref{theo:asos} with $\hat\eta$ replaced by
  $\tilde\eta$. Then
  $n^{1/2}(\thetatmle-\theta)\rightsquigarrow N(0,\sigma^2)$, where
  $\sigma^2=\var\{D_\eta(O)\}$ is the non-parametric efficiency bound.
\end{theorem}

The proof of this theorem is presented in the \href{sm}{Supplementary
Materials}. Broadly, the proof proceeds as follows. First, inclusion of the
covariates $H_Y,H_Z,H_M$ guarantees that the tilting model
$\{\m_\beta,\q_\beta,\vv_\beta:\beta\}$ generates a score spanning the EIF. This
is used to show that the estimator solves the EIF estimating equation. The proof
then proceeds using similar arguments as the proof of Theorem~\ref{theo:asos}
for the one-step estimator, using empirical process theory and leveraging
cross-fitting to avoid entropy conditions on the initial estimators of $\eta$.
Since $D_{\tilde \eta_j}$ depends on the full sample through the estimates of
the parameters $\beta$ of the logistic tilting models, the empirical process
treatment is slightly different from that of Theorem~\ref{theo:asos}. As with
the one-step estimator, the Wald-type confidence interval constructed as
$\thetatmle\pm z_{1-\alpha/2} \cdot (\hat\sigma/\sqrt{n})$ is expected to have
asymptotically correct coverage under the conditions of the theorem.

We note that the TML estimator we developed is only applicable to cases with
a binary intermediate confounder. The case of a continuous intermediate
confounder may also be incorporated by using fluctuation models for continuous
variables such as those previously introduced by~\citet{Diaz12,
diaz2015targeted}. The case of a multivariate $Z$ would require a multivariate
fluctuation model and is the subject of future research.

\section{Numerical studies}\label{sims}
We now turn to evaluating the performance of the proposed one-step and TML
estimators of the interventional (in)direct effects. For simplicity, we assume
a binary intermediate confounder $Z$. We evaluate the estimators with synthetic
data from the following joint distribution of $O$:
\begin{align*}
  W_1 &\sim \text{Bernoulli}\{0.6\}\\
  W_2\mid W_1 &\sim \text{Bernoulli}\{0.3\}\\
  W_3\mid (W_2, W_1) &\sim \text{Bernoulli}\{0.2 + ((W_1 + W_2) / 3)\}\\
  A \mid W &\sim \text{Bernoulli}\{ \expit(0.25 \cdot (W_1 + W_2 + W_3) +
     3 W_1 \cdot W_2 - 2 \}\\
  Z \mid (A, W) &\sim \text{Bernoulli}\{\expit(((W_1 + W_2 + W_3) / 3) - A -
    A W_3 - 0.25)\}\\
  M \mid (Z, A, W) &\sim \text{Bernoulli}\{\expit(W_1 + W_2 + A - Z +
     A \cdot Z - 0.3 A \cdot W_2) \}\\
  Y \mid (M, Z, A, W) &\sim \text{Bernoulli}\{\expit((A - Z + M - A \cdot Z)
    / (W_1 + W_2 + W_3 + 1))\},
\end{align*}
where $\expit(x)=\{1+\exp(x)\}^{-1}$. For each sample size $n \in \{200, 800,
1800, 3200, 5000\}$, we generated $500$ datasets from the above data-generating
mechanism. For each dataset, the two proposed estimators were computed under six
different scenarios: all nuisance parameters $(\m, \g, \allowbreak \q, \e, \rr,
\uu, \vv)$ estimated consistently, and five additional scenarios in which each
of $(\m, \g, \q, \e, \rr)$ were replaced by an inconsistent estimator.
Consistent estimators were constructed via a Super Learner
ensemble~\citep{vdl2007super} that included two variants of the highly adaptive
lasso~\citep{benkeser2016highly} as well as distinct main-terms and
intercept-only logistic regression models. Inconsistent estimation was achieved
via an intercept-only logistic regression model. Our efficient estimators were
computed using our open source \texttt{medoutcon} \texttt{R}
package~\citep{hejazi2020medoutcon}, freely available at
\url{https://github.com/nhejazi/medoutcon}, while nuisance parameter estimates
relied upon the \texttt{sl3}~\citep{coyle2020sl3} and
\texttt{hal9001}~\citep{coyle2019hal9001} \texttt{R} packages, for the Super
Learner and the highly adaptive lasso, respectively.

Figure~\ref{fig:simula} displays the results of our simulation experiments. We
assess the relative performance of the two estimators in terms of several
metrics: absolute bias scaled by $n^{1/2}$;
the mean squared error (MSE), scaled by $n$, relative to the efficiency bound;
and the coverage of 95\% confidence intervals.
Figure~\ref{supp:fig:simula_sm}, in the \href{sm}{Supplementary Materials},
analogously depicts the relative performance of the estimators in terms of
absolute bias and the coverage of 99\% confidence intervals. In terms of scaled
bias, the simulation results corroborate the results of
Lemma~\ref{lemma:robust}. The only case in which the estimators are expected to
be inconsistent is when $\q$ is inconsistent, in which case both $\uu$ and $\vv$
are also inconsistent, thus failing to satisfy the conditions of the lemma.
These simulations also corroborate the results of Theorem~\ref{theo:asos} in the
sense that the scaled bias is expected to vanish or stabilize asymptotically
when all nuisance estimators are estimated consistently. That this appears true
for the case of an inconsistent $\rr$ should be taken as a particularity of this
simulation, not to be expected in general. Likewise, the scaled MSE converges to
the efficiency bound of the model when all nuisance estimators are consistent,
as predicted by Theorem~\ref{theo:asos}. Interestingly, coverage of confidence
intervals close to the nominal level was achieved in many simulation scenarios
--- this property is also only to be expected in general if all nuisance
parameters are consistently estimated. This is possibly due to anti-conservative
estimation of the variance, a phenomenon previously observed for variance
estimators using empirical influence functions (possible solutions have recently
been presented by~\citet{tran2018robust}). Reassuringly, the coverage of the
confidence intervals for the all-consistent case achieves the nominal level, as
would be expected from theory.
\begin{figure}[H]
  \hspace{-4.2em}
  \flushleft
  \centering
  \includegraphics[scale = 0.3]{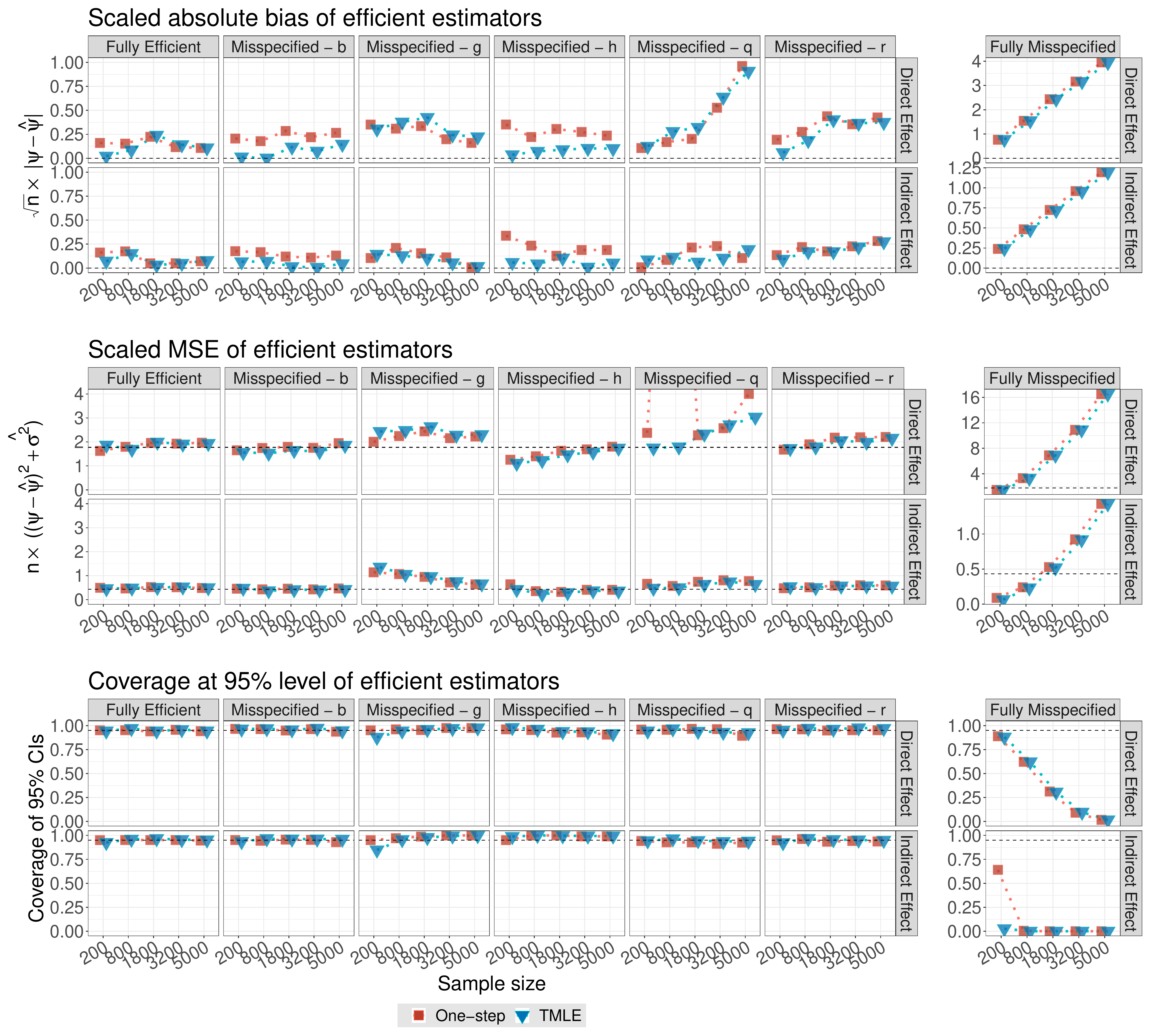}
  \caption{Comparison of efficient estimators across different nuisance
    parameter configurations.}\label{fig:simula}
\end{figure}
Of note, the TML estimator seems to perform uniformly better than the one-step
estimator. For example, in the case where all nuisance parameters are correctly
estimated, in which both estimators have the same asymptotic behavior, the TML
estimator exhibits a superior small-sample bias-variance trade-off. Even in
cases where consistency is unexpected, such as when $\q$ is inconsistently
estimated, the TML estimator offers a better MSE. A particularly interesting
instance in which the one-step estimator performs poorly seems to be the case in
which $\g$ or $\e$ are inconsistent.
This is a somewhat puzzling result, due to possible irregularity of the
estimators in the absence of all-consistent estimation of nuisance functions.
This suboptimal behavior is attenuated for the TML estimator in terms of both
bias and MSE.

\section{Empirical Illustration}\label{application}
\subsection{Overview and set-up}
We now apply our proposed estimator to estimate indirect effects of the
randomized receipt of a Section 8 housing voucher as a young child on subsequent
risky behavior in adolescence, through mediators related to the school
environment, and direct effect not operating through those mediators. Our
illustrative data come from the Moving to Opportunity (MTO) study,
a longitudinal, randomized trial in which families living in high-rise public
housing in five U.S.~cities were randomized to receive a Section 8 housing
voucher or not~\citep{kling2007experimental}. Such a voucher could then be used
to move out of public housing and into a rental on the private market.
Participating families were followed up at two separate times over the
subsequent 10--15 years.

We use baseline information for covariates, $W$, and treatment assignment, $A$.
Covariates include individual and family sociodemographic information,
motivation for participation, neighborhood perceptions, and relevant
interactions. Moving with the voucher out of public housing represents the
intermediate confounder, $Z$. The school environment variables that comprise $M$
were assessed over the intervening 10--15 years between baseline and the final
follow-up visit; these are presented in Table~\ref{tab:examp}. We consider the
outcome of risky behavior measured at the final follow-up, $Y$, measured using
the behavior problems index (BPI), which is calculated based on the proportion
of risky behaviors endorsed~\citep{zill1990behavior}.

We restrict to girls in the Boston, Chicago, and New York City sites ($N=1749$),
as previous work has shown heterogeneities by
sex~\citep{osypuk2012gender,rudolph2018mediation} and study
site~\citep{rudolph2018composition}, and because the total effects for girls in
this subset of sites were similar. We present results from this restricted
analysis using one imputed dataset~\footnote{Multiple imputation by chained
equations~\citep{buuren2010mice} was used to create $30$ imputed datasets,
predicting missing values in the order of the data (i.e., $W, A,Z, M, Y$). The
first imputed dataset was used, as this analysis was for illustrative purposes.}
to address missing data instead of completing a full stratified analysis using
multiple imputed datasets, as our goal is to illustrate the proposed method.
A comprehensive mediation analysis is the subject of future substantive work.

We use machine learning to flexibly and data-adaptively model the components of
Equations (\ref{sc:Y})--(\ref{sc:M}). Specifically, using the Super Learner
ensembling procedure, we build an ensemble prediction algorithm as the convex
combination of algorithms in a user-supplied library in order to minimize the
5-fold cross-validated prediction error~\citep{vdl2007super}. The chosen
algorithm library includes $\ell_1$-penalized
regression~\citep{tibshirani1996regression}, generalized linear models, and
multivariate adaptive regression splines~\citep{friedman1991multivariate}.

\begin{table}[ht]
\centering
\caption{Mediators included by group. }
\label{tab:examp}
\begin{tabular}{|l|p{12cm}|}
  \hline
  Mediator Group    & Mediators Included                                                                                                                                                                                     \\ \hline
  Binary, one       & Attended $>4$ schools                                                                                                                                                                                  \\
  Binary, fewer     & `Binary, one' AND  ever attended a school ranked above the 50\textsuperscript{th} percentile, attended schools where more than 75\% of students received free or reduced price lunch, ever attended a non-Title I school \\
  Continuous, fewer & Average rank of schools attended, average percentage of students receiving free or reduced price lunch, number of schools attended, percentage of schools attended that were Title I                   \\
  Continuous, all   & `Continuous, fewer' AND number of moves since baseline, student to teacher ratio, school attended at follow-up is in a different school district than baseline                                           \\
  \hline
\end{tabular}
\end{table}

\subsection{Results}

We apply the proposed one-step and targeted minimum loss estimators (TMLE) to
estimate the (a) indirect effects of voucher receipt on risky behavior in
adolescence, through various aspects of the school environment, and the (b)
direct effects not operating through the school environment, shown in
Figure~\ref{fig:examp}. Standard errors and confidence intervals are computed
based on the results of Theorems~\ref{theo:asos} and~\ref{theo:astmle}.

Generally, results for the one-step and TML estimators are similar in this
example. The indirect effect through attending more than four schools over the
duration of follow-up is estimated as an increase in
0.113\% (95\%CI: -0.109\%, 0.336\%) by the TML estimator; the indirect effects
estimated including the combination of four binary mediators listed in
Table~\ref{tab:examp} are similar. The indirect effects estimated using the
continuous versions of these four mediators are slightly stronger: (risk
difference (RD): 0.205\%, 95\%CI: -0.038\%-0.448\%). Lastly, we estimated
indirect effects using the combination of seven continuous mediators listed in
Table~\ref{tab:examp}. These indirect effects are stronger still, but with
much wider confidence intervals (RD: 0.459\%, 95\%CI: -0.247\%, 1.165\%).

Thus, the mechanisms we examine related to the school environment were
universally found to unintentionally result in more risky behavior in adolescent
girls whose families were randomized to receive a housing voucher.
This harmful indirect effect is largely attributable to
increased school instability (i.e., increased likelihood of attending more
schools over the intervening 10--15 years).
The corresponding direct effect point estimates are negative, meaning that the
effect of voucher receipt not operating through aspects of the school
environment/instability is estimated to result in less risky behavior in
adolescence. For example, the direct effect of voucher receipt not operating
through the continuous, smaller group of mediators decreases risky behavior by
approximately 3.246\% (95\% CI: -7.264, 0.420).


\begin{table}[ht]
\centering
\caption{Indirect and direct effect estimates and 95\% confidence
     intervals in illustrative example.}
\label{fig:examp}
\begin{tabular}{|l | l | c c|}
   \hline
  Mediator Group    & Estimator & Indirect effect (95\% CI)  & Direct effect  (95\% CI)\\ 
  \hline
Continuous, all   & One-step & 0.37 (-0.33, 1.07) & -3.16 (-7.00, 0.69)  \\
                  & TMLE     & 0.46 (-0.25, 1.17) & -3.45 (-7.31, 0.40)  \\
Continuous, fewer & One-step & 0.21 (-0.03, 0.45) & -3.42 (-7.26, 0.42)  \\
                  & TMLE     & 0.21 (-0.04, 0.45) & -3.26 (-7.10, 0.57)  \\
Binary, fewer     & One-step & 0.11 (-0.14, 0.35) & -3.32 (-7.14, 0.50)  \\
                  & TMLE     & 0.10 (-0.14, 0.35) & -3.18 (-7.00, 0.63) \\
Binary, one       & One-step & 0.16 (-0.06, 0.38) & -3.45 (-7.22, 0.31)  \\
                  & TMLE     & 0.16 (-0.06, 0.38) & -3.45 (-7.21, 0.31)  \\ \hline
\end{tabular}
\end{table}

\section*{Supplementary material}

Supplementary material available online includes proofs of presented results.


\bibliographystyle{biometrika}
\bibliography{refs}
\end{document}


\maketitle


\section{Efficient influence function Theorem~\ref{main:theo:eif}}
\begin{proof}
  In this proof we will use $\Theta(\P)$ to denote a parameter as a
  functional that maps the distribution $\P$ in the model to a real
  number. We will assume that the measure $\nu$ is discrete so that
  integrals can be written as sums. It can be checked algebraically
  that the resulting influence function will also correspond to the
  influence function of a general measure $\nu$. The true parameter
  value is thus given by
  \[\theta=\Theta(\P) = \sum_{y,z,m,w} y\,\p(y\mid a', z, m,
    w)\q(z\mid a',w)\p(m\mid a^{\star},w)\p(w).\] The non-parametric MLE of
  $\theta$ is given by
  \begin{equation}
    \Theta(\Pn)=\sum_{y,z,m,w}y\frac{\Pn f_{y,a',z,m,w}}{\Pn
      f_{a',z,m,w}}\frac{\Pn f_{z,a',w}}{\Pn f_{a',w}}\frac{\Pn
      f_{m,a^{\star},w}}{\Pn f_{a^{\star},w}}\Pn f_w\label{nonpest},
  \end{equation}
  where we remind the reader of the notation $\P f =\int f \dd\P$. Here
  $f_{y,a,z,m,w}=\one(Y=y,A=a,Z=z,M=m,W=w)$, and $\one(\cdot)$ denotes
  the indicator function. The other functions $f$ are defined
  analogously.

  We will use the fact that the efficient influence function in a
  non-parametric model corresponds with the influence curve of the
  NPMLE. This is true because the influence curve of any regular
  estimator is also a gradient, and a non-parametric model has only
  one gradient. The Delta method \cite[see, e.g., Appendix 18
  of][]{vanderLaanRose11} shows that if $\hat \Theta(\Pn)$ is a
  substitution estimator such that $\theta=\hat \Theta(\P)$, and
  $\hat \Theta(\Pn)$ can be written as
  $\hat \Theta^{\star}(\Pn f:f\in\mathcal{F})$ for some class of functions
  $\mathcal{F}$ and some mapping $\Theta^{\star}$, the influence function of
  $\hat \Theta(\Pn)$ is equal to
  \[\dr_\P(O)=\sum_{f\in\mathcal{F}}\frac{\dd\hat
  \Theta^{\star}(\P)}{\dd\P f}\{f(O)-\P f\}.\]

  Applying this result to (\ref{nonpest}) with
  $\mathcal{F}=\{f_{y,a',z,m,w},f_{a',z,m,w},f_{z,a',w},f_{a',w},
  f_{m,a^{\star},w}, f_{a^{\star},w},f_w:
  y,z,m,w\}$ and rearranging terms gives the result of the theorem.
  The algebraic derivations involved here are lengthy and not
  particularly illuminating, and are therefore omitted from the proof.
\end{proof}

\section{Lemma~\ref{main:lemma:aeif}}
\begin{proof}
This result follows by replacing
\[\p(m\mid a^{\star}, w)
  =\p(m\mid a', z, w)\frac{\g(a'\mid w)}{\g(a^{\star}\mid w)} \frac{\q(z\mid
    a',w)}{\rr(z\mid a',m,w)}\frac{\e(a^{\star}\mid m, w)}{\e(a'\mid m,w)}\]
in expression (\ref{main:eq:defhuv}) in the \href{paper}{main paper}.
\end{proof}

\section{Theorem~\ref{main:theo:asos}}
  \begin{proof}
    Let $\Pnj$ denote
    the empirical distribution of the prediction set ${\cal V}_j$, and let
    $\Gnj$ denote the associated empirical process
    $\sqrt{n/J}(\Pnj-\P)$. Let $\Gn$ denote the empirical process
    $\sqrt{n}(\Pn-\P)$. Note that
    \[\thetaos = \frac{1}{J}\sum_{j=1}^J\Pnj D_{\hat
        \eta_j},\,\,\,\theta=\P D_{\eta}.\]Thus,
    \[  \sqrt{n}(\thetaos - \theta)=\Gn (D_{\eta}
      - \theta) + R_{n,1} + R_{n,2},\]
    where
    \[  R_{n,1}  =\frac{1}{\sqrt{J}}\sum_{j=1}^J\Gnj(D_{\hat
        \eta_j} - D_{\eta}),\,\,\,
      R_{n,2}  = \frac{\sqrt{n}}{J}\sum_{j=1}^J\P(D_{\hat
        \eta_j}-\theta).
    \]
    It remains to show that $R_{n,1}$ and $R_{n,2}$ are
    $o_P(1)$. Theorem~\ref{theo:mr} together with the Cauchy-Schwartz
    inequality and assumptions \ref{main:ass:pos} and
    \ref{main:ass:so} of the theorem shows that $R_{n,2}=o_P(1)$. For
    $R_{n,1}$ we use empirical process theory to argue conditional on
    the training sample ${\cal T}_j$. Let
    ${\cal F}_n^j=\{D_{\hat \eta_j} - D_\eta\}$ denote the
    class of functions with one element. Because the function
    $\hat\eta_j$ is fixed given the training data, we can apply
    Theorem 2.14.2 of \cite{vanderVaart&Wellner96} to obtain
    \begin{equation}
      E\left\{\sup_{f\in {\cal F}_n^j}|\Gnj f| \,\,\bigg|\,\, {\cal
          T}_j\right\}\lesssim \lVert F^j_n \rVert \int_0^1
      \sqrt{1+N_{[\,]}(\epsilon \lVert F_n^j \rVert, {\cal F}_n^j,
        L_2(\P))}\dd\epsilon,\label{eq:empproc}
    \end{equation} where
    $N_{[\,]}(\epsilon \lVert F_n^j \rVert, {\cal F}_n^j, L_2(\P))$ is the
    bracketing number and we take
    $F_n^j= \lvert D_{\hat \eta_j} -
    D_\eta \rvert$ as an envelope for the class ${\cal
      F}_n^j$. Theorem 2.7.2 of \cite{vanderVaart&Wellner96} shows
    \[\log N_{[\,]}(\epsilon \lVert F_n^j \rVert, {\cal F}_n^j, L_2(\P))\lesssim
      \frac{1}{\epsilon \lVert F_n^j \rVert}.\]
    This shows
    \begin{align*}
      \lVert F^j_n \rVert \int_0^1\sqrt{1+N_{[\,]}(\epsilon
      \lVert F_n^j \rVert, {\cal F}_n^j, L_2(\P))}\dd\epsilon &\lesssim
                                                                \int_0^1\sqrt{\lVert F^j_n \rVert^2+\frac{\lVert F^j_n
                                                                \rVert}{\epsilon}}\dd\epsilon\\ &\leq
                                                                                                  \lVert F^j_n \rVert + \lVert F^j_n \rVert^{1/2}
                                                                                                  \int_0^1\frac{1}{\epsilon^{1/2}}\dd\epsilon\\
                                                              &\leq \lVert F^j_n \rVert + 2 \lVert F^j_n \rVert^{1/2}.
    \end{align*}
    Since $\lVert F^j_n \rVert = o_P(1)$, this shows
    $\sup_{f\in {\cal F}_n^j}\Gnj f=o_P(1)$ for each $j$, conditional
    on ${\cal T}_j$. Thus $ R_{n,1}=o_P(1)$. Theorem~\ref{theo:mr}
    together with the Cauchy-Schwartz inequality and assumptions
    \ref{main:ass:pos} and \ref{main:ass:so} of the theorem shows that
    $R_{n,2}=o_P(1)$, concluding the proof of the theorem.
  \end{proof}

\section{Theorem~\ref{main:theo:astmle}}
\begin{proof}
  Let $D_{Y,\eta}(o)$ denote (\ref{main:sc:Y}), $D_{Z,\eta}(o)$ denote
  (\ref{main:sc:Z}), and $D_{M,\eta}(o)$ denote (\ref{main:sc:M}). By
  definition, the sum of the scores of the submodels
  $\{\m_\beta,\q_\beta,\vv_\alpha:(\beta,\alpha)\}$ at the last
  iteration of the TMLE procedure is equal to
  $\sum_{i=1}^n \{D_{Y, \tilde \eta}(O_i)+D_{Z, \tilde
    \eta}(O_i)+D_{M, \tilde \eta}(O_i)=0$. Thus, we have
  \[\thetatmle = \frac{1}{J}\sum_{j=1}^J\Pnj D_{\tilde \eta_j}.\]
  Thus,
  \[  \sqrt{n}(\thetatmle - \theta)=\Gn (D_{\eta}
    - \theta) + R_{n,1} + R_{n,2},\]
  where
  \[  R_{n,1}  =\frac{1}{\sqrt{J}}\sum_{j=1}^J\Gnj(D_{\tilde
      \eta_j} - D_{\eta}),\,\,\,
    R_{n,2}  = \frac{\sqrt{n}}{J}\sum_{j=1}^J\P(D_{\tilde
      \eta_j}-\theta).
  \]
  As in the proof of Theorem~\ref{main:theo:asos}, Theorem~\ref{theo:mr}
  together with the Cauchy-Schwartz inequality and assumptions
  \ref{main:ass:pos} and \ref{main:ass:so} of the theorem shows that
  $R_{n,2}=o_P(1)$.

  Since $D_{\tilde \eta_j}$ depends on the full sample through the
  estimates of the parameters $\beta$ of the logistic tilting models,
  the empirical process treatment of $R_{n,1}$ needs to be slightly
  from that in the proof of Theorem~\ref{main:theo:asos}. To make this
  dependence explicit, we introduce the notation
  $D_{\hat \eta_j,\beta}=D_{\tilde \eta_j}$ and $R_{n,1}(\beta)$. Let
  ${\cal F}_n^j=\{D_{\hat \eta_j,\beta} - D_\eta:\beta\in B\}$. Then
  (\ref{eq:empproc}) also holds here with
  $F_n^j=\sup_{\beta\in B} \lvert D_{\hat \eta_j,\beta} - D_\eta
  \rvert$ as an envelope for the class ${\cal F}_n^j$. The rest of the
  steps of the proof can be applied to show
  $\sup_{f\in {\cal F}_n^j}\Gnj f=o_P(1)$ for each $j$, conditional on
  ${\cal T}_j$. Thus $\sup_{\beta\in B}
  R_{n,1}(\beta)=o_P(1)$. Theorem~\ref{theo:mr} together with the
  Cauchy-Schwartz inequality and assumptions \ref{main:ass:pos} and
  \ref{main:ass:so} of the theorem shows that $R_{n,2}=o_P(1)$,
  concluding the proof of the theorem.
\end{proof}

\section{Additional results}

\begin{lemma}[Lemma 19.33 of \cite{vanderVaart98}]
  For any finite class $\mathcal F$ of bounded, measurable,
  square-integrable functions with $|\mathcal F|$ elements,
  \[E\left\{\max_f\big|\Gn f\big|\right\}\lesssim
    \max_f\frac{||f||_\infty}{n^{1/2}}\log (1 + |\mathcal F|) +
  \max_f\lVert f \rVert \{\log (1+|\mathcal F|)\}^{1/2}\]
\end{lemma}

\begin{theorem}[Multiple robustness of the EIF]\label{theo:secondo}
  For notational simplicity, in this theorem we omit the dependence of all
  functions on $w$. Let $\p_W$ denote the distribution of $W$. We have {\small
    \begin{align*}
      \P & D_{\eta_1} -\theta = \notag\\
         &\int
           \{\vv_1(a^{\star})-\vv(a^{\star})\}\left
           \{1-\frac{\g(a^{\star})}{\g_1(a^{\star})}\right\}
           \dd\P_W\dd\nu(m,z) + \\
         & \int \frac{\g(a')}{\g_1(a^{\star})}\left
         \{\frac{\q_1(z\mid a')}{\rr_1(z\mid
           a',m)}\frac{\e_1(a^{\star}\mid m)}{\e_1(a'\mid m)} -
           \frac{\q(z\mid a')}{\rr(z\mid a',m)}\frac{\e(a^{\star}\mid m)}
           {\e(a'\mid m)}\right\}\{\m(a',z,m) - \m_1(a',z,m)\}
           \p(m, z\mid a')\dd\P_W\dd\nu(m,z) + \\
         &\int\frac{\g(a')}{\g_1(a^{\star})}\left\{\uu(z,a') -
         \uu_1(z,a')\right\}\{\q_1(z\mid a')-\q(z\mid
           a')\}\dd\P_W\dd\nu(m,z)+\\
         &\int\frac{\g(a')}{\g_1(a^{\star})}\frac{\q(z\mid
           a')}{\rr(z\mid a',m)}\frac{\e(a^{\star} \mid m)}{\e(a'\mid
           m)}\left\{\m_1(a',z,m) - \m(a',z,m)\right\}\{\q_1(z\mid a')-\q(z\mid
           a')\}\p(m\mid z, a')\dd\P_W\dd\nu(m,z).
    \end{align*}}
\label{theo:mr}
\end{theorem}
\begin{proof}
  For fixed $a'$ and $a^{\star}$ we have
  \begin{align}
    \P & D_{\eta_1} -\theta = \notag\\
       & \int \frac{\g(a')}{\g_1(a^{\star})}\frac{\q_1(z\mid a')}{\rr_1(z\mid
         a',m)}\frac{\e_1(a^{\star}\mid m)}{\e_1(a'\mid m)}\{\m(a',z,m) -
         \m_1(a',z,m)\}\p(m\mid z, a')\q(z\mid a')
         \dd\P_W\dd\nu(m,z)\nu(m,z)\label{eq:t1}\\
    +&\int\frac{\g(a')}{\g_1(a^{\star})}\uu_1(z,a')\{\q(z\mid a') -
       \q_1(z\mid a')\}\dd\P_W\dd\nu(m,z)\label{eq:t2}\\
    +& \int \frac{\g(a^{\star})}{\g_1(a^{\star})}\m_1(a',z,m)
    \q_1(z\mid a')\p(m\mid a^{\star})\dd\P_W\dd\nu(m,z)\label{eq:t3}\\
    -&\int \vv_1(a^{\star})\left\{\frac{\g(a^{\star})}
    {\g_1(a^{\star})}-1\right\}\dd\P_W\dd\nu(m,z)\label{eq:t4}\\
    -&\int \vv(a^{\star})\dd\P_W\dd\nu(m,z).\label{eq:t5}
  \end{align}
  First, note that
  \begin{align}
    (\ref{eq:t4}) + (\ref{eq:t5}) &=\int \{\vv_1(a^{\star})-
    \vv(a^{\star})\}\left\{1-\frac{\g(a^{\star})}{\g_1(a^{\star})}\right\}
    \dd\P_W\dd\nu(m,z)\label{eq:t6}\\
    & - \int \vv(a^{\star})\frac{\g(a^{\star})}{\g_1(a^{\star})}
    \dd\P_W\dd\nu(m,z)\notag.
  \end{align}
  Using (\ref{main:eq:param}) in the \href{paper}{main manuscript}, we get
  \begin{align}\int
    \vv(a^{\star})&\frac{\g(a^{\star})}{\g_1(a^{\star})}\dd\P_W\dd\nu(m,z)
            =\notag\\
            & \int\frac{\g(a')}{\g_1(a^{\star})}\frac{\q(z\mid a')}{\rr(z\mid
             a',m)}\frac{\e(a^{\star}\mid m)}{\e(a'\mid m)}\{\m(a',z,m)-
             \m_1(a',z,m)\}\p(m\mid z, a')\q(z\mid a')\dd\P_W\dd\nu(m,z)+
             \notag\\
            &\int\frac{\g(a')}{\g_1(a^{\star})}\frac{\q(z\mid a')}{\rr(z\mid
            a',m)}\frac{\e(a^{\star}\mid m)}{\e(a'\mid m)}\m_1(a',z,m)\p(m\mid
            z, a')\q(z\mid a')\dd\P_W\dd\nu(m,z)\label{eq:t9}
  \end{align}
  Thus
  \begin{align}
    (\ref{eq:t1})  + (\ref{eq:t4}) + (\ref{eq:t5}) &= (\ref{eq:t6})\notag\\
      &+ \int \frac{\g(a')}{\g_1(a^{\star})}\left\{\frac{\q_1(z\mid a')}
      {\rr_1(z\mid a',m)}\frac{\e_1(a^{\star} \mid m)}{\e_1(a'\mid m)} -
      \frac{\q(z\mid a')}{\rr(z\mid a',m)}\frac{\e(a^{\star}\mid m)}{\e(a'\mid
      m)}\right\}\times\notag\\ &\,\,\,\,\,\,\,\,\{\m(a',z,m) -
      \m_1(a',z,m)\}\p(m\mid z, a')\q(z\mid a')
      \dd\P_W\dd\nu(m,z)\label{eq:t12}\\ &-(\ref{eq:t9})\notag.
  \end{align}
  Substituting (\ref{main:eq:param}) in the \href{paper}{main manuscript} into
  (\ref{eq:t3}) we get
  \[(\ref{eq:t3})-(\ref{eq:t9})=\int\frac{\g(a')}{\g_1(a^{\star})}\frac{\q(z\mid
      a')}{\rr(z\mid a',m)}\frac{\e(a^{\star}\mid m)}{\e(a'\mid
      m)}\m_1(a',z,m)\p(m\mid z, a')\{\q_1(z\mid a')-\q(z\mid
    a')\}\dd\P_W\dd\nu(m,z),\]
  which yields
  {\small
  \begin{align}
    (\ref{eq:t2})&+(\ref{eq:t3})-(\ref{eq:t9})\notag\\
                 &=\int\frac{\g(a')}{\g_1(a^{\star})}\left\{\frac{\q(z\mid
                   a')}{\rr(z\mid a',m)}\frac{\e(a^{\star}\mid m)}{\e(a'\mid
                   m)}\m_1(a',z,m)\p(m\mid z, a') -
                   \uu_1(z,a')\right\}\{\q_1(z\mid a')-\q(z\mid
                   a')\}\dd\P_W\dd\nu(m,z)\notag\\
                   &=\int\frac{\g(a')}{\g_1(a^{\star})}\left\{\uu(z,a') -
                  \uu_1(z,a')\right\}\{\q_1(z\mid a')-\q(z\mid
                   a')\}\dd\P_W\dd\nu(m,z)\label{eq:t10}\\
                  &+\int\frac{\g(a')}{\g_1(a^{\star})}\frac{\q(z\mid
                   a')}{\rr(z\mid a',m)}\frac{\e(a^{\star}\mid m)}{\e(a'\mid
                   m)}\left\{\m_1(a',z,m) - \m(a',z,m)\right\}\{\q_1(z\mid a')-\q(z\mid
                   a')\}\p(m\mid z, a')\dd\P_W\dd\nu(m,z)\label{eq:t11}
  \end{align}}
  Putting everything together yields
  \[(\ref{eq:t1}) +(\ref{eq:t2}) + (\ref{eq:t3}) + (\ref{eq:t4}) +
    (\ref{eq:t5})=(\ref{eq:t6}) +(\ref{eq:t12})+(\ref{eq:t10})
    +(\ref{eq:t11}),\]
  yielding the result of the theorem.
\end{proof}

\section{Parameterizations if $Z$ is high-dimensional}
In this section we present a parameterization of the efficient
influence function that may be used to construct a one-step estimator
if the intermediate confounder is high-dimensional but the mediator
$M$ is not. For notational simplicity in this section we override the
definitions of $\q$, $\rr$, $\vv$, and $\uu$ given in the main
manuscript. Let $\q(m\mid a, w)$, $\rr(m\mid z, a, w)$ denote the
corresponding conditional densities of $M$. Define
the density ratio
\[\ddd(a, z, m, w) = \frac{\p(z\mid a,w)}{\p(z\mid a,m,w)}=\frac{\e(a\mid m,w)}{\e(a^{\star}\mid m, w)}\frac{\g(a^{\star}\mid w)}{\g(a\mid w)}\frac{\q(m\mid a^{\star}, w)}{\rr(m\mid a, z,
    w)},\]
and the regression functions
\begin{align*}
  \vv(a,w) &= \E\left\{\int_{\mathcal
             M}\m(a',Z,m,W)\q(m\mid a^{\star},W)\dd\nu(m),\bigg|\,
             A=a,W=w\right\},\\
  \uu(m,a,w) &= \E\left\{\m(A,Z,M,W)\ddd(A,Z,M,W)\,\bigg|\, M=m,A=a,W=w\right\}.
\end{align*}

\begin{lemma}[Alternative representation of the EIF for
  high-dimensional $Z$]\label{lemma:aeif}
  The efficient influence function for $\theta$ admits the following alternative representation:
  \begin{align*}
    D_\P(o) &= \frac{\one\{a=a'\}}{\g(a'\mid w)}\h(a',z,m,w)\{y - \m(a',z,m,w)\}\\
            & + \frac{\one\{a=a'\}}{\g(a'\mid w)}\left\{\int_{\mathcal
              M}\m(a',z,m,w)\q(m\mid a^{\star},w)\dd\nu(m)-\vv(a',w)\right\}\\
            & + \frac{\one\{a=a^{\star}\}}{\g(a^{\star}\mid
              w)}\left\{\uu(m,a',w)-\int_{\mathcal
              M}\uu(m,a',w)\q(m\mid a^*,w)\dd\nu(m)\right\}\\
            & + \vv(a',w).
  \end{align*}
\end{lemma}
With the definitions given in this section, the efficient influence
function satisfies the following multiple robustness lemma:
\begin{lemma}[Multiple robustness of $D_\eta(O)$]\label{lemma:robust1}
  Let $\eta_1=(\m_1,\g_1,\e_1,\q_1,\rr_1,\uu_1,\vv_1)$ be such that one
  of the following conditions hold:
  \begin{enumerate}[label=(\roman*)]
  \item $\vv_1=\vv$ and either $(\q_1,\rr_1)=(\q,\rr)$
    or $(\m_1,\q_1)=(\m,\q)$ or $(\m_1,\uu_1)=(\m,\uu)$, or
  \item $\g_1=\g$ and either $(\q_1,\rr_1)=(\q,\rr)$
    or $(\m_1,\q_1)=(\m,\q)$ or $(\m_1,\uu_1)=(\m,\uu)$.
  \end{enumerate}
  Then $\P D_{\eta_1} = \theta$.
\end{lemma}
The proof of this lemma is analogous to the proof of
Theorem~\ref{theo:secondo} above. Here we note that the condition
$\e_1=\e$, though not explicitly stated, is part of the lemma because
it is necessary for $\uu_1=\uu$. Data-adaptive regression and density
estimation may be used to obtain cross-fitted estimates
$\hat \eta_{j}$ of $\eta=(\m,\g,\e,\q,\rr,\uu,\vv)$, as outlined in
the main manuscript.  The one-step estimator is defined as the
solution to the cross-fitted efficient influence function estimating
equation:
\begin{equation*}
  \thetaos = \frac{1}{n} \sum_{i = 1}^n D_{\hat\eta_{j(i)}}(O_i).
\end{equation*}
With the definitions given in this section, this estimator also
satisfies a weak convergence theorem:
\begin{theorem}[Weak convergence of the one-step estimator]\label{theo:asos} Let
  $\lVert \cdot \rVert$ denote the $L_2(\P)$ norm defined as
  $\lVert f \rVert^2=\int f^2\dd\P$. Assume
  \begin{enumerate}[label=(\roman*)]
  \item \label{ass:pos} Positivity of exposure and confounder probabilities:
    \[\P\{\g(a^{\star} \mid W) > \epsilon\}=\P\{\rr(M\mid a',Z,W) >
      \epsilon\}=\P\{\e(a^{\star}\mid M,W) > \epsilon\}=1\] for some
    $\epsilon > 0$.
  \item \label{ass:so} $n^{1/2}$-convergence of second order terms:
    \[\lVert \hat\vv-\vv \rVert\, \lVert \hat\g-\g \rVert +
      \lVert \hat\m-\m \rVert \{\lVert \hat\q-\q \rVert + \lVert
      \hat\rr - \rr \rVert\} + \lVert \hat\uu - \uu \rVert\, \lVert
      \hat\q-\q \rVert = o_\P(n^{-1/2}).\]
  \end{enumerate}
  Then we have
  \[n^{1/2}(\thetaos-\theta)\rightsquigarrow N(0,\sigma^2),\] where
  $\sigma^2=\var\{D_\eta(O)\}$ is the non-parametric efficiency bound.
\end{theorem}
\begin{proof}
  The proof of this theorem is identical to the proof of
  Theorem~\ref{main:theo:asos} and is ommitted.
\end{proof}

\section{Relation with the methods of \citet{zheng2017longitudinal}}\label{remark:wen}
   \citet{zheng2017longitudinal} propose the effect
  decomposition
  \begin{equation}
    \E_c(Y_{a', \Gamma_{a'}} - Y_{a^{\star}, \Gamma_{a^{\star}}}) =
    \underbrace{\E_c(Y_{a', \Gamma_{a'}} - Y_{a',
        \Gamma_{a^{\star}}})}_{\text{Indirect effect}} + \underbrace{\E_c(Y_{a',
        \Gamma_{a^{\star}}} - Y_{a^{\star}, \Gamma_{a^{\star}}})}_{\text{Direct
        effect}},\label{eq:wen}
  \end{equation} where $\Gamma_a$ is a random draw from the
  distribution of $M_a$ conditional on $(Z_a,W)$. This
  indirect effect does not capture the path
  $A \rightarrow Z \rightarrow M \rightarrow Y$. Thus, it does not provide
  an analogous effect decomposition to the natural direct and indirect
  effects. To see why, consider the DAG in Figure (\ref{fig:dag2}), in
  which we have an exclusion restriction deleting the paths
  $A\rightarrow M$ and $A\rightarrow Y$. Such a restriction is reasonable when
  $A$ is randomization status and $Z$ is exposure uptake as in our illustrative
  application.
  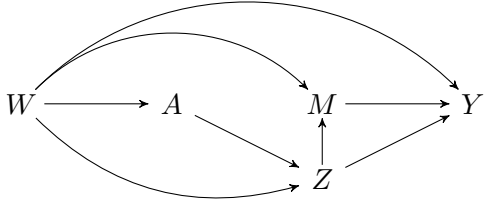
\begin{figure}[!htb]
    \centering
    \begin{tikzpicture}
      \Vertex{0, -1}{Z}
      \Vertex{-4, 0}{W}
      \Vertex{0, 0}{M}
      \Vertex{-2, 0}{A}
      \Vertex{2, 0}{Y}
      \ArrowR{W}{Z}{black}
      \Arrow{Z}{M}{black}
      \Arrow{M}{Y}{black}
      \Arrow{A}{Z}{black}
      \Arrow{Z}{Y}{black}
      \Arrow{W}{A}{black}
      \ArrowL{W}{Y}{black}
      \ArrowL{W}{M}{black}
    \end{tikzpicture}
    \caption{Directed Acyclic Graph (\ref{fig:dag}) with an exclusion
      restriction.}
    \label{fig:dag2}
  \end{figure}
  In this graph, the effect mediated by $M$ is only through the path
  $A \rightarrow Z \rightarrow M \rightarrow Y$. For simplicity,
  assume independence of all exogenous variables $U$. We have
  \begin{align}
    \Ps(M_a\mid Z_a,W) &= \Ps(M_a, Z_a\mid W)/\Ps(Z_a\mid
                         W)\notag\\
                       &= \P(M, Z\mid A=a,W)/\P(Z\mid
                         A=a,W)\label{eq:nuc}\\
                       &= \P(M\mid Z, A=a,W)\notag\\
                       &= \P(M\mid Z,W)\label{eq:byas},
  \end{align}
  where (\ref{eq:nuc}) follows because $(M_a,Z_a)\indep A\mid W$, and
  (\ref{eq:byas}) follows because $M\indep A\mid (Z,W)$. Thus,
  $\Gamma_{a'}$ has the same distribution as $\Gamma_{a^{\star}}$, and the
  ``indirect effect'' in (\ref{eq:wen}) is always zero, even when
  there is an effect through the path
  $A \rightarrow Z \rightarrow M \rightarrow Y$. Likewise, the
  ``direct effect'' in (\ref{eq:wen}) is equal to the total effect,
  and it measures the paths $A\rightarrow Z\rightarrow Y$ and
  $A \rightarrow Z \rightarrow M \rightarrow Y$.

\section{Further results of numerical studies}\label{more_sims}
Previously, we presented the results of numerical experiments designed to
evaluate our proposed efficient estimators with synthetic data produced by the
data-generating mechanism given in section~\ref{main:sims}. We examined standard
metrics for evaluating our estimators, including the absolute bias scaled by
$n^{1/2}$; the mean squared error (MSE), scaled by $n$, relative to the
efficiency bound; and the coverage of 95\% confidence intervals. For brevity, we
elected not to present their relative performance in terms of (unscaled)
absolute bias and the coverage of 99\% confidence intervals. Using results
produced by the same set of numerical experiments, Figure~\ref{fig:simula_sm}
provides an illustration of the relative performance of our estimators in terms
of these metrics.
\begin{figure}[H]
  \hspace{-4.2em}
  \includegraphics[scale = 0.37]{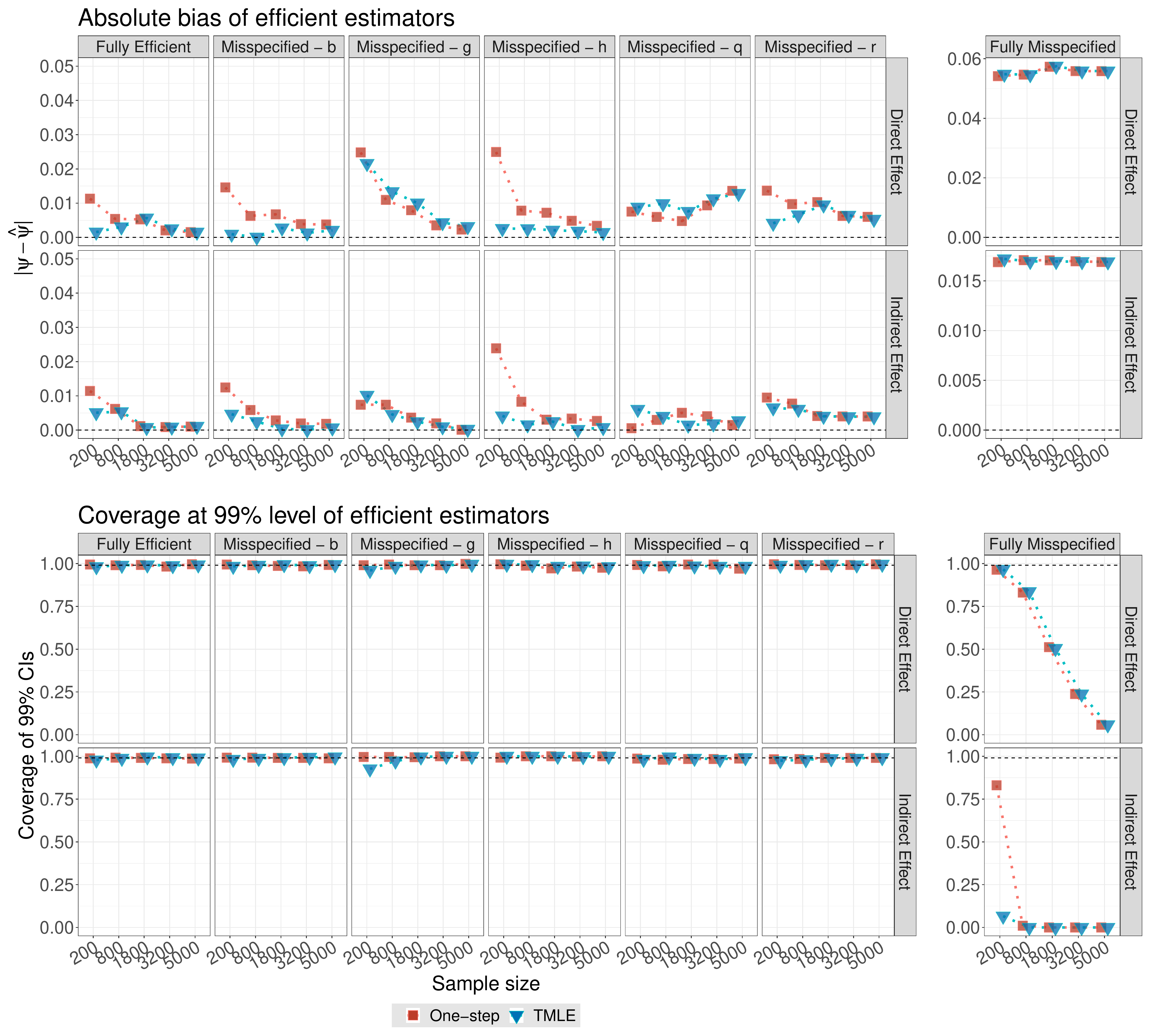}
  \caption{Results of numerical simulation comparing efficient estimators
    across different configurations of nuisance
    parameters.}\label{fig:simula_sm}
\end{figure}
As with the results presented in Figure~\ref{main:fig:simula}, evaluation of our
estimators with respect to their absolute bias and coverage of 99\% confidence
intervals corroborates the results of both Lemma~\ref{main:lemma:robust} and
Theorem~\ref{main:theo:asos}. We recall from our discussion in
section~\ref{main:sims} that the estimators are expected to be inconsistent only
when $\q$ is inconsistently estimated (n.b., both $\uu$ and $\vv$ are also
inconsistent in this case), in which case the conditions of
Lemma~\ref{main:lemma:robust} remain unsatisfied. In this scenario, neither of
our estimators displays a bias that is asymptotically vanishing. Though both
appear to provide better performance in estimating the indirect effect than the
direct effect, this is likely only a peculiarity of the data-generating
mechanism considered in this example, not a phenomenon expected to recur widely.
As previously discussed with respect to the scaled absolute bias, the relatively
small (and vanishing) bias of the estimators under inconsistent estimation of
$\rr$ is another unusual particularity of this simulation example. In all other
cases, the bias is very small or converges to zero, excepting the fully
misspecified case, of course. In all settings other than the fully misspecified
case, the coverage of the 99\% confidence intervals appears to attain the
nominal level. Again, this is rather unexpected --- that is, this is only to be
expected when all nuisance parameters are consistently estimated --- though such
performance could also arise from anti-conservative estimation of the variance,
which would also yield considerably inflated coverage of confidence intervals.
Finally, as expected from theory, the coverage of the 99\% confidence intervals
for the fully efficient case achieves the nominal level.

\bibliographystyle{plainnat}
\bibliography{refs}